\documentclass[english,10pt,conference]{IEEEtran}
\usepackage[T1]{fontenc}
\usepackage[latin9]{inputenc}
\usepackage[letterpaper]{geometry}
\usepackage{verbatim}
\usepackage{amsmath}
\usepackage{graphicx}
\usepackage{amssymb}

\makeatletter
\newtheorem{definitn}{Definition}
\newtheorem{thm}{Theorem}
\newtheorem{cor}{Corollary}
\newtheorem{remrk}{Remark}
\newtheorem{example}{Example}

\usepackage{tikz}

\DeclareMathOperator*{\argmin}{arg\,min}
\DeclareMathOperator*{\argmax}{arg\,max}
\date{}

\IEEEoverridecommandlockouts

\makeatother

\usepackage{babel}

\begin{document}

\title{On the Joint Decoding of LDPC Codes and Finite-State Channels via
Linear Programming}

\author{\authorblockN{Byung-Hak Kim and Henry D. Pfister%
\thanks{This material is based upon work supported by the National Science
Foundation under Grant No. 07407470. Any opinions, findings, conclusions,
or recommendations expressed in this material are those of the authors
and do not necessarily reflect the views of the National Science Foundation.%
}} \authorblockA{Department of Electrical and Computer Engineering,
Texas A\&M University\\
 Email: \{bhkim,hpfister\}@tamu.edu} }
\maketitle
\begin{abstract}
In this paper, the linear programming (LP) decoder for binary linear
codes, introduced by Feldman, et al. is extended to joint-decoding
of binary-input finite-state channels. In particular, we provide a
rigorous definition of LP joint-decoding pseudo-codewords (JD-PCWs)
that enables evaluation of the pairwise error probability between
codewords and JD-PCWs. This leads naturally to a provable upper bound
on decoder failure probability. If the channel is a finite-state intersymbol
interference channel, then the LP joint decoder also has the maximum-likelihood
(ML) certificate property and all integer valued solutions are codewords.
In this case, the performance loss relative to ML decoding can be
explained completely by fractional valued JD-PCWs. 
\end{abstract}

\section{Introduction \label{sec:Intro} \vspace{-1mm}}

\subsection{Motivation}

Message-passing iterative decoding has been a very popular decoding
algorithm in research and practice for the past fifteen years \cite{Wiberg-96}.
In the last five years, linear programming (LP) decoding has been
a popular topic in coding theory and has given new insight into the
analysis of iterative decoding algorithms and their modes of failure
\cite{Feldman-2003}\cite{Feldman-it05}\cite{Pseudocodewords-web}.
For both decoders, fractional vectors, known as pseudo-codewords (PCWs),
play an important role in the performance characterization of these
decoders \cite{Feldman-it05}\cite{Vontobel-itsub07}. This is in
contrast to classical coding theory where the performance of most
decoding algorithms (e.g., maximum-likelihood (ML) decoding) is completely
characterized by the set of codewords. 

For channels with memory, such as finite-state channels (FSCs), the
situation is a bit more complicated. In the past, one typically separated
channel decoding (i.e., estimating the channel inputs from the channel
outputs) from error-correcting code (ECC) decoding (i.e., estimating
the transmitted codeword from estimates of the channel inputs) \cite{Muller-it04}.
The advent of message-passing iterative decoding enabled the joint-decoding
(JD) of the channel and code by iterating between these two decoders
\cite{Douillard-ett95}. 

In this paper, we extend the LP decoder to the JD of binary-input
FSCs and define LP joint-decoding pseudo-codewords (JD-PCWs). This
leads naturally to a provable upper bound (e.g., a union bound) on
the probability of decoder failure as a sum over all codewords and
JD-PCWs. This extension has been considered as a challenging open
problem in the prior work \cite{Tadashi-arxiv07}\cite{Feldman-2003}.
The problem is well posed by Feldman in his PhD thesis \cite[Section 9.5 page 146]{Feldman-2003},
\begin{quote}
\emph{\textquotedbl{}In practice, channels are generally not memoryless
due to physical effects in the communication channel.'' ... {}``Even
coming up with a proper linear cost function for an LP to use in these
channels is an interesting question. The notions of pseudocodeword
and fractional distance would also need to be reconsidered for this
setting.\textquotedbl{} }
\end{quote}
Other than providing satisfying answer to the above open question,
our primary motivation is the prediction of the error rate for joint
decoding at high SNR. The idea is to run a simulation at low SNR and
keep track of all observed codeword and pseudo-codeword errors. A
truncated union bound is computed by summing over all observed errors
and the result is an estimate of the error rate at high SNR. Computing
this bound is complicated by the fact that the loss of channel symmetry
implies that the dominant PCWs may depend on the transmitted sequence.

While we were preparing this manuscript, we became aware of a more
general approach by Flanagan \cite{Flanagan-aller08}\cite{Flanagan-arxiv09}.
In fact, our LP formulation was developed independently but is identical
to his {}``Efficient LP relaxation''. Our motivation, however, is
somewhat different. The main goal is to use the error rate of joint
LP decoding as a tool to analyze joint iterative decoding of FSCs
and low-density parity-check (LDPC) codes. Thus, we give novel prediction
results in Sec. \ref{sec:Simulations}. We also observe that both
formulations provide an ML edge-path certificate that is not equivalent
to an ML codeword certificate (see Remark \ref{rem:MLedgepath} and
\ref{rem:MLcertificate}). This property is not guaranteed by Wadayama's
approach based on quadratic programming \cite{Tadashi-arxiv07}. 

The paper is structured as follows. After briefly reviewing LP decoding
and FSCs in the remainder of Sec. \ref{sec:Intro}, we describe the
LP joint decoder in Sec. \ref{sec:Joint-LPD} and define JD-PCWs in
Sec. \ref{sec:JD-PCW}. In Sec. \ref{sec:Bound}, we discuss the decoder
performance analysis via the union bound (and pairwise error probability)
over JD-PCWs and notions of generalized Euclidean distance. Experimental
results are given in Sec. \ref{sec:Simulations} and conclusions are
given in Sec. \ref{sec:Concl}.\vspace{-1mm}

\subsection{Background }

Feldman, et al. introduced the LP decoding for binary linear codes
in \cite{Feldman-it05}\cite{Feldman-2003}. It is is based on solving
an LP relaxation of an integer program which is equivalent to ML decoding.
Later this method was extended to codes over larger alphabets \cite{Flanagan-it09}
and to the simplified decoding of intersymbol interference (ISI) \cite{Taghavi-itsub07}.
For long codes, the performance of LP decoding is slightly inferior
to iterative decoding but, unlike the iterative decoder, the LP decoder
either detects a failure or outputs a codeword which is guaranteed
to be the ML codeword.

Let $\mathcal{C}\subseteq\left\{ 0,1\right\} ^{n}$ be the length-$n$
binary linear code defined by the parity-check matrix $H$ and $\mathbf{c}=(c_{1},\ldots,c_{n})$
be a codeword. If $\mathcal{I}$ is the set whose elements are the
sets of indices involved in each parity check, then we have\[
\mathcal{C}=\left\{ \mathbf{c}\in\left\{ 0,1\right\} ^{n}\,\bigg|\,\sum_{i\in I}c_{i}\equiv0\,\bmod2,\,\forall\, I\in\mathcal{I}\right\} .\]
The \emph{codeword polytope} is the convex hull of $\mathcal{C}$.
This polytope can be quite complicated to describe though, so instead
one constructs a simpler polytope using local constraints. Each parity-check
$I\in\mathcal{I}$ defines a local constraint that can also be viewed
as a polytope in $\left[0,1\right]^{n}$.
\begin{definitn}
\label{def:LCP} The \emph{local codeword polytope} $\mbox{LCP(\ensuremath{I})}$
associated with a parity check is the convex hull of the bit sequences
that satisfy the check. It is given explicitly by\[
\mbox{LCP}(I)\triangleq\!\!\bigcap_{\substack{S\subseteq I\\
\left|S\right|\text{odd}}
}\left\{ \mathbf{c}\in[0,\,1]^{n}\,\bigg|\sum_{i\in S}c_{i}-\!\!\!\sum_{i\in I-S}\!\! c_{i}\leq\left|S\right|\!-\!1\right\} .\]

\end{definitn}

\begin{definitn}
The \emph{relaxed polytope} $\mathcal{P}(H)$ is the intersection
of the LCPs over all checks, so \begin{align*}
\mathcal{P}(H) & \triangleq\bigcap_{I\in\mathcal{I}}\mbox{LCP}(I).\end{align*}
\end{definitn}
\begin{thm}
[\cite{Feldman-2003}] Consider $n$ consecutive uses of a symmetric
channel $\Pr\left(Y=y|C=c\right)$. If a uniform random codeword is
transmitted and $\mathbf{y}=(y_{1},\ldots,y_{n})$ is received, then
the LP decoder outputs $\mathbf{f}=(f_{1},\ldots,f_{n})$ given by
\[
\argmin_{\mathbf{f}\in\mathcal{P}(H)}\sum_{i=1}^{n}f_{i}\,\mbox{log}\left(\frac{\mbox{Pr}(Y_{i}=y_{i}\,|\, C_{i}=0)}{\mbox{Pr}(Y_{i}=y_{i}\,|\, C_{i}=1)}\right),\]
which is the ML solution if $\mathbf{f}$ is integral (i.e., $\mathbf{f}\in\left\{ 0,1\right\} ^{n}$). \end{thm}
\begin{definitn}
An \emph{LP decoding pseudo-codeword} (LPD-PCW) of a code defined
by the parity-check matrix $H$ is any vertex of the relaxed (fundamental)
polytope $\mathcal{P}(H).$ 
\end{definitn}

\begin{definitn}
A \emph{finite-state channel} (FSC) defines a probabilistic mapping
from a sequence of inputs to a sequence of outputs. Each output $Y_{i}\in\mathcal{Y}$
depends only on the current input $X_{i}\in\mathcal{X}$ and channel
state $S_{i}\in\mathcal{S}$ instead of the entire history of inputs
and channel states. Mathematically, we have $P\left(y,s'|x,s\right)\triangleq\mbox{Pr}\left(Y_{i}\!=\! y,S_{i+1}\!=\! s'|X_{i}\!=\! x,S_{i}\!=\! s\right)$
for all $i$, and we use the shorthand $P\left(y_{1}^{n},s_{2}^{n+1}|x_{1}^{n},s_{1}\right)$
for\begin{multline*}
\mbox{Pr}\left(Y_{1}^{n}\!=\! y_{1}^{n},S_{2}^{n+1}\!=\! s_{2}^{n+1}|X_{1}^{n}\!=\! x_{1}^{n},S_{1}\!=\! s_{1}\right)\\
=\prod_{i=1}^{n}P\left(y_{i},s_{i+1}|x_{i},s_{i}\right).\end{multline*}

\end{definitn}

\begin{definitn}
\label{def:FSISI}A \emph{finite-state intersymbol interference channel}
(FSISIC) is a FSC whose next state is a deterministic function, $\eta(x,s)$,
of the current state $s$ and input $x$. Mathematically, this implies
that\[
\sum_{y\in\mathcal{Y}}P\left(y,s'|x,s\right)=\begin{cases}
1 & \mbox{if}\,\eta(x,s)=s'\\
0 & \mbox{otherwise}\end{cases}.\]

\end{definitn}

\begin{definitn}
The \emph{dicode channel} (DIC) is a binary-input FSISI channel with
a linear response of $G(z)=1-z^{-1}$ and Gaussian noise. If the input
bits are differentially encoded prior to transmission, then the resulting
channel is called the \emph{precoded dicode channel} (pDIC). The state
diagrams of these two channels are shown in Fig. \ref{fig:dic}. \vspace{0mm}
\end{definitn}
\begin{figure}[t]
\begin{centering}
\includegraphics[scale=0.5]{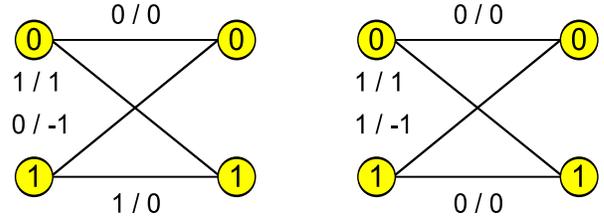}\vspace{-2mm}
\par\end{centering}

\caption{\label{fig:dic}State diagrams for noiseless dicode channel with and
without precoding. The edges are labeled by the input/output pair.}

\end{figure}

\section{New Results: LP Joint-Decoding \label{sec:Joint-LPD}}

Now, we describe \textbf{the LP joint decoder} in terms of the trellis
of the FSC and the checks in the binary linear code. Let $n$ be the
length of the code and $\mathbf{y}$ be the received sequence. The
trellis consists of $(n+1)|\mathcal{S}|$ vertices (i.e., one for
each state and time) and a set $\mathcal{E}$ of at most $2n|\mathcal{S}|^{2}$
edges (not $2n|\mathcal{S}|$, i.e., one edge for each input-labeled
state transition and time). For each edge $e\in\mathcal{E}$, the
functions $t(e)\rightarrow\{1,\ldots,n\}$,~$s(e)\rightarrow\mathcal{S}$,
$s'(e)\rightarrow\mathcal{S}$,~$x(e)\rightarrow\{0,1\}$,~and $a(e)\rightarrow\mathcal{A}$
map this edge to its respective time index, initial state, final state,
input bit, and noiseless output symbol. The LP formulation requires
one variable for each edge $e\in\mathcal{E}$, and we denote that
variable by $g(e)$. Likewise, the LP decoder requires one cost variable
for each edge and we use the branch metric \begin{align*}
b(e)\negthinspace\negthinspace\triangleq\negthinspace\negthinspace\begin{cases}
\negthinspace\negthinspace-\mbox{ln}P\negthinspace\left(y_{t(e)},s'(e)|x(e),s(e)\right)\negthinspace\negthinspace\negthinspace\negthinspace & \mbox{if}\, t(e)>1\\
\negthinspace\negthinspace-\mbox{ln}\left(P\negthinspace\left(y_{t(e)},s'(e)|x(e),s(e)\right)\negthinspace P\negthinspace\left(s(e)\right)\right)\negthinspace\negthinspace & \mbox{if}\, t(e)=1\end{cases}.\end{align*}
 
\begin{definitn}
The \emph{trellis polytope} $\mathcal{T}$ enforces the flow conservation
constraints for channel decoder. The flow constraint for state $j$
at time $i$ is given by\[
\mathcal{F}_{ij}\triangleq\left\{ g(\cdot)\in[0,1]^{\left|\mathcal{E}\right|}\left|\,\sum_{\substack{e\in\mathcal{E}:\\
t(e)=i,\\
s'(e)=j}
}g(e)=\!\!\!\!\sum_{\substack{e\in\mathcal{E}:\\
t(e)=i+1,\\
s(e)=j}
}g(e)\right.\right\} .\]
Using this, the \emph{trellis polytope} $\mathcal{T}$ is given by\[
\mathcal{T}\triangleq\left\{ g(\cdot)\in\bigcap_{\substack{i=1,\ldots n-1\\
j\in\mathcal{S}}
}\mathcal{F}_{ij}\left|\sum_{\substack{e\in\mathcal{E}:\\
t(e)=1}
}g(e)=1\right.\right\} .\]
\end{definitn}
\begin{thm}
[\cite{Feldman-2003}] \label{thm:TLP} Finding the ML edge-path through
a weighted trellis is equivalent to solving the minimum-cost flow
LP \[
\argmin_{g(\cdot)\in\mathcal{T}}\sum_{e\in\mathcal{E}}g(e)b(e)\]
and the optimum $g(\cdot)$ must be integral (i.e., $g(\cdot)\in\left\{ 0,1\right\} ^{\left|\mathcal{E}\right|}$)
unless there are ties.\end{thm}
\begin{definitn}
\label{def:Projection}Let \emph{$\mathcal{Q}$ }be the projection
of $g(\cdot)$ onto the input vector $\mathbf{f}=\left(f_{1},\ldots,f_{n}\right)\in[0,1]^{n}$
where $\mathbf{f}=\mathcal{Q}\mathbf{g}$ and \[
f_{i}=\sum_{\substack{e\in\mathcal{E}:\, t(e)=i,\, x(e)=1}
}g(e).\]

\end{definitn}

\begin{definitn}
\label{def:TPOLY}The \emph{trellis-wise relaxed polytope} $\mathcal{P}_{\mathcal{T}}(H)$
for $\mathcal{P}(H)$ is defined by \vspace{-1mm}

\[
\mathcal{P}_{\mathcal{T}}(H)\triangleq\left\{ g(\cdot)\in\mathcal{T}\left|\mathcal{Q}g\in\mathcal{P}(H)\right.\right\} .\]

\end{definitn}

\begin{definitn}
\label{prop:TCWPOLY}The \emph{set of trellis-wise codewords} $\mathcal{C}_{\mathcal{T}}$
for $\mathcal{C}$ is defined as\vspace{-1mm} \[
\mathcal{C}_{\mathcal{T}}\triangleq\left\{ g(\cdot)\in\mathcal{P}_{\mathcal{T}}(H)\left|g(\cdot)\in\left\{ 0,1\right\} ^{\left|\mathcal{E}\right|}\right.\right\} .\]
\end{definitn}
\begin{thm}
\label{thm:JointLP}\textbf{The LP joint decoder} computes\[
\argmin_{g(\cdot)\in\mathcal{P}_{\mathcal{T}}(H)}\sum_{e\in\mathcal{E}}g(e)b(e)\]
and outputs a joint ML edge-path if $g(\cdot)$ is integral.\end{thm}
\begin{proof}
Let $\mathcal{V}$ be the set of valid input/state sequence pairs.
For a given $\mathbf{y}$, the ML edge-path decoder computes \vspace{-1mm}\begin{align*}
 & \argmax_{(x_{1}^{n},s_{1}^{n+1})\in\mathcal{V}}P(y_{1}^{n},s_{2}^{n+1}|x_{1}^{n},s_{1})P(s_{1})\\
 & =\argmax_{g(\cdot)\in\mathcal{C}_{\mathcal{T}}}P(s_{1})\prod_{\substack{e\in\mathcal{E}:\, g(e)=1}
}P\left(y_{t(e)},s'(e)|x(e),s(e)\right)\\
 & =\argmin_{g(\cdot)\in\mathcal{C}_{\mathcal{T}}}\sum_{\substack{e\in\mathcal{E}:\, g(e)=1}
}b(e)\\
 & =\argmin_{g(\cdot)\in\mathcal{C}_{\mathcal{T}}}\sum_{e\in\mathcal{E}}g(e)b(e),\end{align*}
where ties are resolved in a systematic manner and $b(e)$ at $t(e)=1$
has an additional initial state term as $\mbox{-ln}P\left(s(e)\right)$.
By relaxing $\mathcal{C}_{\mathcal{T}}$ into $\mathcal{P}_{\mathcal{T}}(H)$,
we obtain the desired result.\end{proof}
\begin{cor}
\label{cor:JointLPISI}For a FSISIC, the LP joint decoder outputs
a joint ML codeword if $g(\cdot)$ is integral.\end{cor}
\begin{proof}
The joint ML decoder for codewords computes\vspace{-1mm}\begin{align*}
 & \argmax_{x_{1}^{n}\in\mathcal{C}}\sum_{s_{2}^{n+1}\in\mathcal{S}^{n}}P(y_{1}^{n},s_{2}^{n+1}|x_{1}^{n},s_{1})P(s_{1})\\
 & =\argmax_{x_{1}^{n}\in\mathcal{C}}\sum_{s_{2}^{n+1}\in\mathcal{S}^{n}}\prod_{i=1}^{n}P(y_{i},s_{i+1}|x_{i},s_{i})P(s_{1})\\
 & \stackrel{(a)}{=}\argmax_{x_{1}^{n}\in\mathcal{C}}\prod_{i=1}^{n}P\left(y_{i},\eta\left(x_{i},s_{i}\right)\big|x_{i},s_{i}\right)P(s_{1})\\
 & \stackrel{(b)}{=}\argmin_{g(\cdot)\in\mathcal{C}_{\mathcal{T}}}\sum_{e\in\mathcal{E}}g(e)b(e),\end{align*}
where $(a)$ follows from Defn. \ref{def:FSISI} and $(b)$ holds
because each input sequence defines a unique edge-path. Therefore,
the LP joint-decoder outputs an ML codeword if $g(\cdot)$ is integral.\end{proof}
\begin{remrk}
\label{rem:MLedgepath} If the channel is not a FSISIC (e.g., finite-state
fading channel), the integer valued solutions of the LP joint-decoder
are ML edge-paths and not necessarily ML codewords. This occurs because
the decoder is unable to sum to the probability of the multiple edge-paths
associated with the same codeword (e.g., if multiple distinct edge-paths
are associated with the same input labels).\vspace{-2mm}
\end{remrk}

\section{Joint-Decoding Pseudo-codewords \label{sec:JD-PCW}}

Pseudo-codewords have been observed and given names by a number of
authors \cite{Wiberg-96}\cite{Di-it02}\cite{Richardson-aller03},
but the simplest general definition was provided by Feldman, et al.
in the context of LP decoding of parity-check codes \cite{Feldman-it05}.
One nice property of the LP decoder is that it always returns an integer
codeword or a fractional pseudo-codeword. Vontobel and Koetter have
shown that a very similar set of pseudo-codewords also affect message-passing
decoders, and that they are essentially fractional codewords that
cannot be distinguished from codewords using only local constraints
\cite{Vontobel-itsub07}. We define JD-PCW this section because of
their primary importance in the characterization of code performance
at very low error rates. 
\begin{definitn}
\label{def:TCW}The output of the LP joint decoder is a\textbf{ }\emph{trellis-wise
(ML) codeword} (TCW) if $g(e)\in\{0,\,1\}$ for all $e\in\mathcal{E}$.
Otherwise, if $g(e)\in(0,\,1)$ for some $e\in\mathcal{E}$, then
the solution is called a \emph{joint-decoding trellis-wise pseudo-codeword}
(JD-TPCW) and the decoder outputs {}``failure''.
\end{definitn}

\begin{definitn}
\label{def:SCW}Any TCW $\mathbf{g}$ can be projected onto a \emph{(symbol-wise)
codeword} (SCW)%
\begin{comment}
in many-to-one manner (depends on the initial state assumption)
\end{comment}
{} $\mathbf{f}=\mathcal{Q}\mathbf{g}$. Likewise, any JD-TPCW $\mathbf{g}$
can be projected onto a \emph{joint-decoding symbolwise pseudo-codeword}
(JD-SPCW) %
\begin{comment}
in one-to-many manner for $\mathbf{f}=\left(f_{1},\, f_{2},\,\ldots\,,\, f_{n}\right)$
where
\end{comment}
{} $\mathbf{f}=\mathcal{Q}\mathbf{g}.$ \end{definitn}
\begin{remrk}
\label{rem:MLcertificate}For FSISIC, the LP joint decoder has the
\emph{ML certificate }property; if the decoder outputs a SCW, then
it is guaranteed to be the ML codeword (see Cor. \ref{cor:JointLPISI}).
\end{remrk}

\begin{definitn}
Any TCW can be projected onto a \emph{symbol-wise signal-space codeword}
(SSCW) %
\begin{comment}
in many-to-one manner (depends on the initial state assumption)
\end{comment}
{} and any JD-TPCW $\mathbf{g}$ can be projected onto a\emph{ joint-decoding
symbol-wise signal-space pseudo-codeword} (JD-SSPCW) $\mathbf{p}=\left(p_{1},\ldots,p_{n}\right)$
by averaging the components with \vspace{-1mm} \[
p_{i}=\sum_{\substack{e\in\mathcal{E}:\, t(e)=i}
}g(e)a(e).\]

\end{definitn}

\begin{example}
Consider the single parity-check code SPC(3,2). Over precoded dicode
channel (starts in zero state) with AWGN, this code has five joint-decoding
pseudo-codewords. A simulation was performed for joint-decoding of
the SPC(3,2) on the pDIC trellis and the set of JD-TPCW, by ordering
the trellis edges appropriately, was found to be \begin{align*}
\{(0\,1\,0\,0;0\,0\,.5\,.5;0\,.5\,.5\,0), & (.5\,.5\,0\,0;.5\,0\,0\,.5;0\,1\,0\,0),\\
(.5\,.5\,0\,0;0\,.5\,.5\,0;0\,0\,1\,0), & (1\,0\,0\,0;.5\,.5\,0\,0;0\,.5\,.5\,0),\\
(.5\,.5\,0\,0;.5\,0\,0\,0;0\,.5\,.5\,0)\}.\end{align*}
Using $\mathcal{Q}$ to project them into $\mathcal{P}(H)$, we get
the corresponding set of JD-SPCW \[
\{(1,.5,.5),\,(.5,.5,1),\,(.5,.5,0),\,(0,.5,.5),\,(.5,0,.5)\}.\]
\vspace{-2mm}
\end{example}

\section{Union Bound for LP Joint-Decoding \label{sec:Bound}}

Now that we have defined the relevant pseudo-codewords, we turn our
attentions to the question of {}``how bad'' a certain pseudo-codeword
is, i.e., we want to quantify pairwise error probabilities. In fact,
we will use the insights gained in the previous section to obtain
a union bound on the decoder word error probability (as a tight approximation)
to analyze the performance of the proposed LP-joint decoder. Toward
this end, let's consider the pairwise error event between a SSCW $\mathbf{c}$
and a JD-SSPCW $\mathbf{p}$ first.
\begin{thm}
\label{thm:PEE}A necessary and sufficient condition for the pairwise
decoding error between a SSCW $\mathbf{c}$ and a JD-SSPCW $\mathbf{p}$
is\vspace{-1mm} \[
\sum_{\substack{e\in\mathcal{E}}
}g(e)b(e)\leq\sum_{\substack{e\in\mathcal{E}}
}\tilde{g}(e)b(e),\]
\vspace{-1mm}where $g(\cdot)\in\mathcal{P}_{\mathcal{T}}(H)$ and
$\tilde{g}(\cdot)\in\mathcal{C}_{\mathcal{T}}$ are the LP variables
for $\mathbf{p}$ and $\mathbf{c}$ respectively. 
\end{thm}
%
\begin{comment}
Since the condition in Thm. \ref{thm:PEE} is difficult to handle
directly, we are currently working on to obtain the general expression
that breaks the error probability into PCW of the code and error events
in the channel trellis. 
\end{comment}
{}For the moment, let $\mathbf{c}$ be the SSCW of FSISIC to an AWGN
channel whose output sequence is $\mathbf{y}=\mathbf{c}+\mathbf{v}$,
where $\mathbf{v}=(v_{1},\ldots,v_{n})$ is an i.i.d. Gaussian sequence
with mean $0$ and variance $\sigma^{2}$. We will show that each
pairwise probability has a simple closed-form expression that depends
only on a generalized squared Euclidean distance $d_{gen}^{2}\left(\mathbf{c},\,\mathbf{p}\right)$
and the noise variance $\sigma^{2}.$ The next few definitions and
theorems can be seen as a generalization of \cite{Forney-ima99} and
a special case of the more general formulation in \cite{Flanagan-arxiv09}.
\begin{thm}
\label{thm:LPAWGN} Let $\mathbf{y}$ be the output of a FSISIC with
zero-mean AWGN whose variance is $\sigma^{2}$ per output. Then, the
LP joint decoder is equivalent to \[
\argmin_{g(\cdot)\in\mathcal{P}_{\mathcal{T}}(H)}\sum_{e\in\mathcal{E}}g(e)\left(y_{t(e)}-a(e)\right)^{2}.\]
\end{thm}
\begin{proof}
For each edge $e\in\mathcal{E}$, the output $y_{t(e)}$ is Gaussian
with mean $a(e)$ and variance $\sigma^{2}$, so we have $P\left(y_{t(e)},s'(e)|x(e),s(e)\right)\sim\mathcal{N}\left(a(e),\,\sigma^{2}\right)$.
Therefore, the LP joint-decoder computes\vspace{-1mm}\begin{align*}
 & \argmin_{g(\cdot)\in\mathcal{P}_{\mathcal{T}}(H)}\sum_{e\in\mathcal{E}}g(e)b(e)=\!\!\argmin_{g(\cdot)\in\mathcal{P}_{\mathcal{T}}(H)}\sum_{e\in\mathcal{E}}g(e)\left(y_{t(e)}\!-\! a(e)\right)^{2}\!.\end{align*}

\end{proof}
%
\begin{comment}
\begin{proof}
\begin{align*}
 & \argmin_{g(\cdot)\in\mathcal{P}_{\mathcal{T}}(H)}\left(-\sum_{e\in\mathcal{E}}g(e)\,\mbox{ln}\, P\left(y_{t(e)},s'(e)|x(e),s(e)\right)\right)\\
 & =\argmin_{g(\cdot)\in\mathcal{P}_{\mathcal{T}}(H)}\sum_{e\in\mathcal{E}}g(e)\left(y_{t(e)}-a(e)\right)^{2}.\end{align*}

\end{proof}

\end{comment}
{}
\begin{definitn}
\label{def:dist}Let $\mathbf{c}$ be a SSCW and $\mathbf{p}$ a JD-SSPCW.
Then the \emph{generalized squared Euclidean distance} between $\mathbf{c}$
and $\mathbf{p}$ can be defined in terms of their trellis-wise descriptions
by \[
d_{gen}^{2}\left(\mathbf{c},\,\mathbf{p}\right)\triangleq\frac{\left(\left\Vert \mathbf{d}\right\Vert ^{2}+\sigma_{p}^{2}\right)^{2}}{\left\Vert \mathbf{d}\right\Vert ^{2}}\]
\vspace{-1mm}where \begin{align*}
\left\Vert \mathbf{d}\right\Vert ^{2} & \triangleq\sum_{i=1}^{n}\left(c_{i}-p_{i}\right)^{2},\,\sigma_{p}^{2}\triangleq\sum_{\substack{e\in\mathcal{E}}
}g(e)a^{2}(e)-\sum_{i=1}^{n}p_{i}^{2}.\end{align*}
\vspace{-1mm}\end{definitn}
\begin{thm}
\label{thm:PEP}The pairwise error probability between a SSCW $\mathbf{c}$
and a JD-SSPCW $\mathbf{p}$ is\vspace{-1mm} \[
\mbox{Pr}\left(\mathbf{c}\rightarrow\mathbf{p}\right)=Q\left(\frac{d_{gen}\left(\mathbf{c},\,\mathbf{p}\right)}{2\sigma}\right).\]
\vspace{-1mm}\end{thm}
\begin{proof}
The pairwise error probability $\mbox{Pr}\left(\mathbf{c}\rightarrow\mathbf{p}\right)$
that the LP joint-decoder will choose the pseudo-codeword $\mathbf{p}$
over $\mathbf{c}$ can be written as $\mbox{Pr}\left(\mathbf{c}\rightarrow\mathbf{p}\right)=$\vspace{-1mm}\begin{align*}
 & \mbox{Pr}\left\{ \sum_{i=1}^{n}\sum_{\substack{e\in\mathcal{E}:\\
t(e)=i}
}g(e)\left(y_{t(e)}-a(e)\right)^{2}\leq\sum_{i=1}^{n}\left(y_{i}-c_{i}\right)^{2}\right\} \\
 & =\mbox{Pr}\left\{ \begin{array}{c}
\sum_{i}y_{i}\left(c_{i}-p_{i}\right)\leq\frac{1}{2}\left(\sum_{i}c_{i}^{2}-\sum_{\substack{e\in\mathcal{E}}
}g(e)a^{2}(e)\right)\end{array}\right\} \\
 & \stackrel{(a)}{=}Q\left(\frac{\sum_{i}c_{i}\left(c_{i}-p_{i}\right)-\frac{1}{2}\left(\sum_{i}c_{i}^{2}-\sum_{\substack{e\in\mathcal{E}}
}g(e)a^{2}(e)\right)}{\sigma\sqrt{\sum_{i}\left(c_{i}-p_{i}\right)^{2}}}\right)\\
 & \stackrel{(b)}{=}Q\left(\frac{\left\Vert \mathbf{d}\right\Vert ^{2}+\sigma_{p}^{2}}{2\sigma\left\Vert \mathbf{d}\right\Vert }\right)\stackrel{(c)}{=}Q\left(\frac{d_{gen}\left(\mathbf{c},\,\mathbf{p}\right)}{2\sigma}\right),\end{align*}
where $(a)$ follows from the fact that $\sum_{i}y_{i}\left(c_{i}-p_{i}\right)$
is distributed $\mathcal{N}\left(\sum_{i}c_{i}(c_{i}-p_{i}),\,\sum_{i}(c_{i}-p_{i})^{2}\right)$
and $(b)/(c)$ follow from Defn. \ref{def:dist}.
\end{proof}
Wiberg was the first to define a generalization of the Euclidean distance
to explain errors caused by iterative decoding \cite{Wiberg-96} and
this was extended to non-binary cases \cite{Forney-ima99} where $\mbox{Pr}\left(\mathbf{c}\rightarrow\mathbf{p}\right)$
looks very similar to Thm. \ref{thm:PEP}. The main difference is
the definition of the trellis-wise approach used for JD-TPCW. 

The performance degradation of LP decoding relative to ML decoding
can be explained by pseudo-codewords and their contribution to the
error rate depends on $d_{gen}\left(\mathbf{c},\,\mathbf{p}\right).$
Indeed, by defining $K_{d_{gen}}(\mathbf{c})$ as the number of codewords
and JD-PCWs at distance $d_{gen}$ from $\mathbf{c}$ and $\mathcal{G}(\mathbf{c})$
as the set of generalized Euclidean distances, we can write the union
bound on word error rate (WER) as\vspace{-1mm} \[
P_{w|\mathbf{c}}\leq\sum_{d_{gen}\in\mathcal{G}(\mathbf{c})}K_{d_{gen}}(\mathbf{c})\, Q\left(\frac{d_{gen}}{2\sigma}\right).\]
\vspace{-1mm}Of course, we need the set of JD-TPCWs to compute $\mbox{Pr}\left(\mathbf{c}\rightarrow\mathbf{p}\right)$
with the Thm. \ref{thm:PEP}. There are two complications with this
approach. One is that like original problem \cite{Feldman-2003},
no method is known yet for computing the generalized Euclidean distance
spectrum, apart from going through all error events explicitly. Another
is, unlike original problem, the constraint polytope may not be symmetric
under codeword exchange. Therefore the decoder performance may not
be symmetric under codeword exchange. Hence, the decoder performance
may depend on the transmitted codeword. In this case, the pseudo-codewords
will also depend on the transmitted sequence.

\section{Simulation Results and Error Rate Prediction \label{sec:Simulations} }

\vspace{-2mm}In this section, we present simulation results for two
LDPC codes on the precoded dicode channel (pDIC) and use those results
to predict the error rate well beyond the limits of our simulations.
Both codes are $(3,\,5)$-regular binary LDPC codes; the first has
length 155 and the second has length 455. The parity-check matrices
were chosen randomly except that double-edges and four-cycles were
avoided. Since the performance depends on the transmitted codeword,
the results were obtained for 3 randomly chosen codewords of fixed
weight. The weight was chosen to be roughly half the block length,
giving weight 74 in the first case and 226 in the second case.%
\begin{figure*}[t]
\begin{centering}
\includegraphics[scale=0.38]{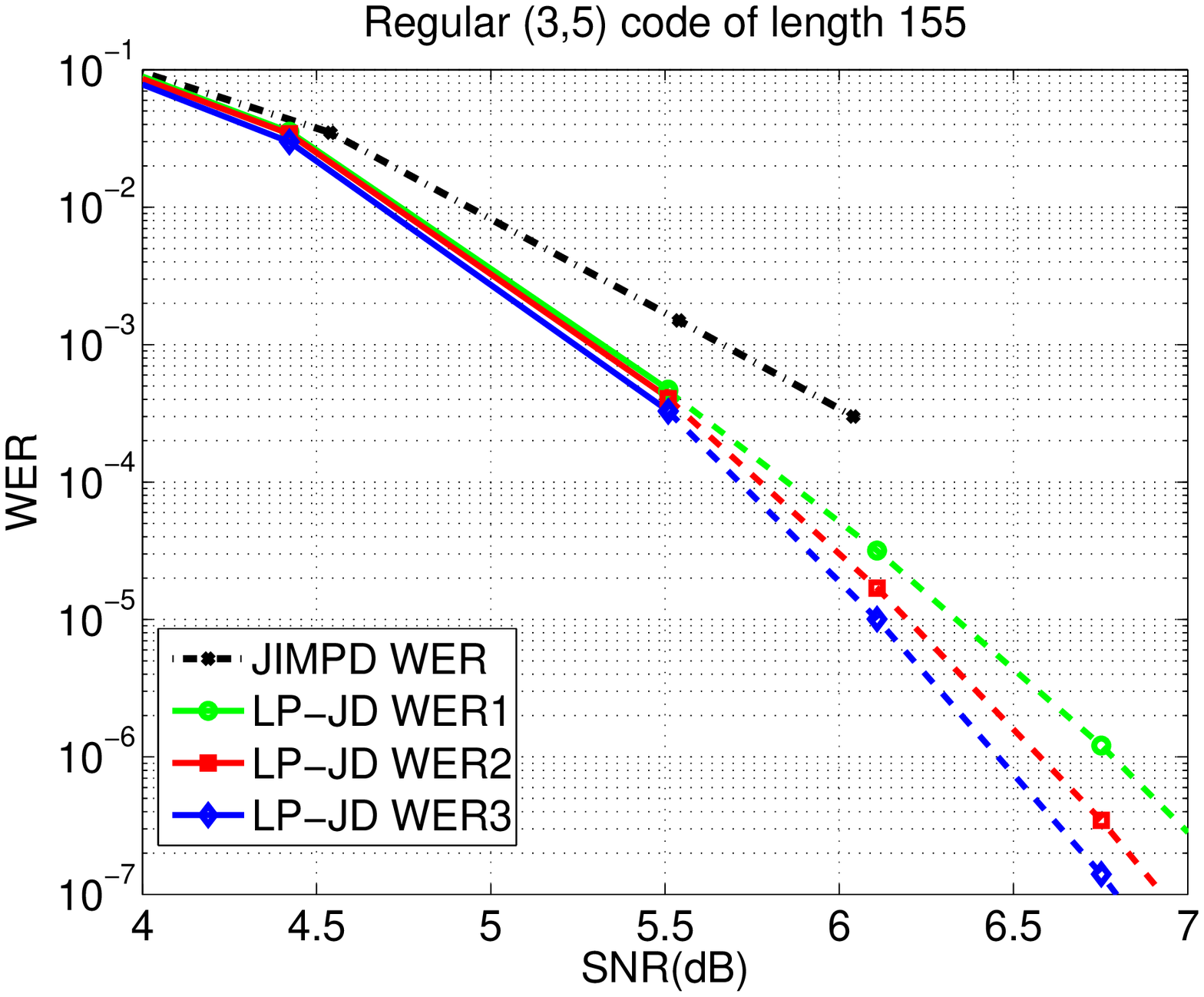}~~~~~~~~~\includegraphics[scale=0.38]{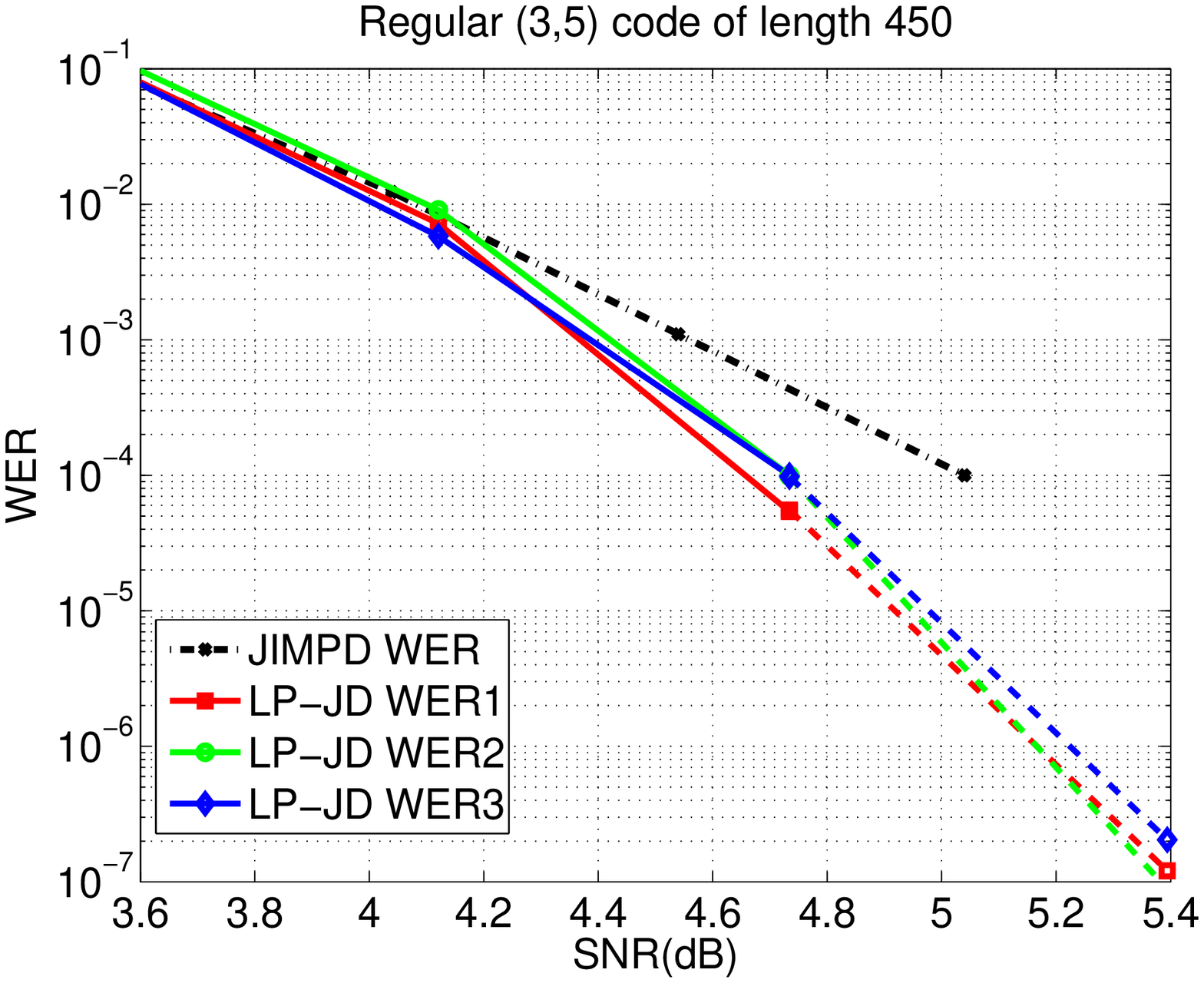}\vspace{-3mm}
\par\end{centering}

\caption{\label{fig:result}This figure shows comparison between the LP joint
decoding and joint iterative message-passing decoding on the precoded
dicode channel with AWGN for random (3,5) regular LDPC codes of length
$n=155$ (left) and $n=450$ (right). The curves shown are the LP-JD
WER (solid), LP-JD WER prediction (dashed) and JIMPD WER (dash-dot).
The experiments were repeated for three different non-zero codewords
in each case. The dashed curves are computed using the union bound
in Sec. \ref{sec:Bound} based on JD-PCWs observed at 3.4558 dB (left)
2.6696 dB (right) and the dash-dot curves are obtained using the state-based
joint iterative message-passing decoder (JIMPD) described in \cite{Kavcic-it03}.
Note that SNR is defined as channel output power divided by $\sigma^{2}$.}
\vspace{-3mm}
\end{figure*}

The results are shown in Fig. \ref{fig:result}. The solid lines represent
the simulation curves while the dashed lines represent a truncated
union bound. The truncated union bound is obtained by computing the
generalized Euclidean distances associated with all decoding errors
that occurred at some low SNR points (e.g., WER of roughly than $10^{-1}$)
until we observe a stationary generalized Euclidean distance spectrum.
This high WER allows the decoder to rapidly discover JD-PCWs. The
dash-dot curves show the state-based joint iterative message-passing
decoder (JIMPD) algorithm described in \cite{Kavcic-it03}. Somewhat
surprisingly, we find that LP joint-decoding outperforms JIMPD by
about 0.5dB at WER of $10^{-4}$.

The LP decoding is performed in the dual domain because this is much
faster than the primal when using MATLAB. Due to the slow speed of
LP decoding still, simulations were completed up to a WER of roughly
$10^{-5}$. It is well-known that the truncated bound should be relatively
tight at high SNR if all the dominant JD-PCWs have been found.

The final complication that must be discussed is the dependence on
the transmitted codeword. It is known that long LDPC codes with joint
iterative decoding experience a concentration phenomenon \cite{Kavcic-it03}
whereby the error probability associated with transmitting a randomly
chosen codeword is very close, with high probability, to the average
error probability over all transmitted codewords. We note that this
effect starts to appear even at the short block lengths used in this
example. More research is required to understand this effect at moderate
block lengths and to verify the same effect for LP decoding. \vspace{-3mm}

\section{Conclusions \label{sec:Concl}}

\vspace{-2mm}In this paper, we present a novel linear-programing
(LP) formulation of joint decoding for LDPC codes on FSCs that offer
decoding performance improvements over joint iterative decoding. Joint-decoding
pseudo-codewords (JD-PCWs) are also defined and the decoder error
rate is upper bounded by a union bound sum over JD-PCWs. %
\begin{comment}
These results are closely related to earlier results by Flanagan \cite{Flanagan-arxiv09},
which we discovered only during the preparation of this manuscript. 
\end{comment}
{}Finally, we propose a simulation-based semi-analytic method for estimating
the error rate of LDPC codes on FSISIC at high SNR using only simulations
at low SNR.\vspace{-3mm}

\bibliographystyle{IEEEtran}
%\addcontentsline{toc}{section}{\refname}\bibliography{IEEEabrv,WCLabrv,WCLbib,WCLnewbib}

\end{document}